\documentclass[11pt]{article}
\usepackage[letterpaper]{geometry}
\usepackage{amsmath,amsthm,amsfonts,amssymb}
\usepackage{enumerate,color,xcolor}
\usepackage{graphicx}
\usepackage{subfigure}
\usepackage{url}
\usepackage{comment}
\numberwithin{equation}{section}
\theoremstyle{plain}

\newtheorem{theorem}{Theorem}

\newtheorem{corollary}[theorem]{Corollary}

\newtheorem{proposition}[theorem]{Proposition}

\newtheorem{remark}[theorem]{Remark}

\begin{document}

\begin{center}
  \Large \bf Asymptotics for the Laplace transform of the time integral of the geometric Brownian motion
\end{center}

\author{}
\begin{center}
Dan Pirjol\,
\footnote{School of Business, Stevens Institute of Technology, Hoboken, NJ 07030, United States of America; dpirjol@gmail.com},
Lingjiong Zhu\,
\footnote{Department of Mathematics, Florida State University, 1017 Academic Way, Tallahassee, FL-32306, United States of America; zhu@math.fsu.edu}
\end{center}

\begin{center}
 \today
\end{center}

\begin{abstract}
We present an asymptotic result for the Laplace transform of the time integral of
the geometric Brownian motion $F(\theta,T) = \mathbb{E}[e^{-\theta X_T}]$
with $X_T = \int_0^T e^{\sigma W_s + ( a - \frac12 \sigma^2)s} ds$, which is
exact in the limit $\sigma^2 T \to 0$ at fixed $\sigma^2 \theta T^2$ and $aT$.
This asymptotic result is applied to pricing zero coupon bonds in the Dothan
model of stochastic interest rates. The asymptotic result provides an approximation 
for bond prices which is in good agreement with numerical evaluations in a wide 
range of model parameters.
As a side result we obtain the asymptotics for Asian option prices in the Black-Scholes model, taking into account interest rates and dividend yield contributions in the 
$\sigma^{2}T\to 0$ limit.
\end{abstract}

\section{Introduction}

The Laplace transform $F(\theta,T)=\mathbb{E}[e^{-\theta X_T}]$ of the
time-integral of the geometric Brownian motion $X_T=\int_0^T e^{\sigma W_t + (a-\frac12 \sigma^2)t} dt$ appears in many problems of applied probability and
mathematical finance. This expectation gives the prices of zero coupon bonds
in the Dothan model of stochastic interest rates \cite{Dothan}
\begin{eqnarray}\label{ZCBdef}
P_{0,T} = 
\mathbb{E}\left[e^{- \int_0^T r_s ds}\right]\,.
\end{eqnarray}
The Dothan model is a short rate model which assumes that the short rate 
$r_t$ follows a geometric Brownian motion (gBM) $r_t = r_0e^{\sigma W_t +(a- \frac12 \sigma^2) t}$ in the risk-neutral measure $\mathbb{Q}$ .
The Dothan model can be regarded as the continuous time limit of the Black-Derman-Toy model \cite{BDT} which is a discrete time model where the one-period interest rate is a geometric Brownian motion sampled at the start of the period.

The expectation (\ref{ZCBdef}) appears also in credit risk, in default intensity models where the default of a company is modeled as the arrival of a Poisson process with intensity following a geometric Brownian motion. In these models the expectation $P_{0,T}$ denotes the survival probability up to time $T$, conditional on survival up to time 0.
 
The time-integral of the asset price $A_T := \int_0^T S_t dt$ plays an
important role in Asian options pricing where it determines the payoff of these options. In particular, in the Black-Scholes model, 
the asset price $S_t=S_0 e^{\sigma W_t + (r-q-\frac12 \sigma^2)t}$ follows a gBM, such that 
$A_T = S_0 X_T$ with $a=r-q$ 
(see \cite{DufresneReview} for a survey). This time-integral also appears in the statistical mechanics of disordered media \cite{CMY}. 

The evaluation of the distribution of $X_T$ and of its Laplace transform 
has received a great deal of attention in the literature. 
An explicit expression for the distribution of $X_T$ was given by \cite{Yor}. However, direct evaluation of the expectation (\ref{ZCBdef})
using the distribution of $X_T$ given in \cite{Yor} is numerically inefficient. 
Several alternative computational methods have been proposed for the numerical
evaluation of the Laplace transform. 

(1). The Feynman-Kac PDE method. This method uses the fact that the Laplace  transform satisfies a parabolic PDE. 
This has been solved by Dothan in \cite{Dothan}.
The result for non-zero drift was corrected in \cite{DothanRev}. 

(2). Monte Carlo methods. A probabilistic representation for the Laplace
transform $F(\theta,T)$ which is more amenable to MC numerical evaluation
was given in \cite{DothanRev}. Its evaluation was studied in \cite{PU}.
An importance sampling MC simulation method using a change of measure
determined by large deviations theory was given recently in \cite{Kim},
using a method proposed in \cite{PG,GuasoniRob}.

The paper makes two main novel contributions.
First, in Section~\ref{sec:main} we give an analytical result for the Laplace transform of the time-integral of the gBM in a certain asymptotic limit $\sigma^2 T\to 0$ at fixed combinations of model parameters \eqref{limit}. 
The result is given by the solution of a variational problem which was studied in a different context in \cite{AsianDiscrete}. However, the new result is not simply a consequence of the result in \cite{AsianDiscrete} since the setting is different (continuous-time vs discrete-time average of the gBM). 
This result has practical applications to bond pricing in the Dothan model, with non-zero drift, 
which are explored in detail in Section~\ref{sec:numerical}.
Second, in Section~\ref{sec:LDP} we obtain an extension of the small-maturity asymptotics for Asian option prices with continuous-time averaging presented in \cite{SIFIN}, which allows for finite interest rates effects. This is the continuous-time counterpart of a result obtained in \cite{AsianDiscrete} for Asian options with discrete-time averaging.
This result has applications to pricing Asian options in the Black-Scholes model
with non-zero interest rates and dividends. Numerical study in \cite{AsianDiscrete} shows that the effects of the interest rates can be significant and their inclusion improves considerably agreement with exact (numerical) computations. Our results extend these asymptotic
results to Asian options with continuous-time averaging.

Finally, the theoretical analysis in this paper relies on the large deviations method \cite{Dembo},
which has been used in similar contexts in the previous literature \cite{SIFIN,AsianDiscrete,SABRdiscrete}. 
The rate function for the large deviations in our context can be expressed
as a variational problem which does not have a simple closed form in general. This
technical challenge makes practical implementation of the asymptotic results less efficient. Our main theoretical contribution is to show that in the context
of time-integral of gBM (with drift), one can still solve this variational problem analytically, and thus extends the existing results in \cite{SIFIN}.

\section{Main result}\label{sec:main}

We prove here an asymptotic result for the Laplace transform:
\begin{equation}\label{eqn:B:T}
F(\theta,T) = 
\mathbb{E}\left[e^{-\theta \int_{0}^{T}e^{\sigma W_{s}+(a-\frac{1}{2}\sigma^{2})s}ds}\right]
\end{equation}
in a particular limit of the model parameters defined by taking 
$\sigma^2 T \to 0$ at fixed
\begin{equation}\label{limit}
b^2 := \frac12 \sigma^2 \theta T^2\,,\quad \zeta := a T \,.
\end{equation}

The zero coupon bond price in the Dothan model (\ref{ZCBdef}) corresponds to identifying $\theta\mapsto r_0$. This limit covers several cases of practical relevance
including small-volatility $\sigma$ at fixed maturity $T$ and large interest rate 
$r_0$,  and it also covers the case of short-maturity $T$ at fixed volatility $\sigma$, and large interest rate $r_0$. In the context of credit risk modeling, $\theta$ corresponds to the initial intensity of default, which can become large for distressed companies.
(When the interest rate $r_{0}$ is small, the analysis is much simpler,
and for the sake of completeness, we include this regime in Appendix A.
We also discuss the large-maturity $T$ regime in Appendix B.)

Our main result is the following theorem.

\begin{theorem}\label{Thm:1}
Consider the $\sigma^2 T\to 0$ limit with the constraint \eqref{limit}.
Then,
\begin{equation}
\lim_{\sigma^2 T\rightarrow 0}(\sigma^{2}T)\log F(\theta,T)
= - J_B(b,\zeta),
\end{equation}
where
\begin{equation}\label{JB}
J_B(b,\zeta) := 
\inf_{h\in\mathcal{AC}[0,1]: h(0)=0}\left\{2 b^2
\int_{0}^{1}e^{h(t)}dt
+\frac{1}{2}\int_{0}^{1}\left(h'(t) - \zeta \right)^{2}dt\right\},
\end{equation}
where $\mathcal{AC}[0,1]$ denotes the space of absolutely continuous functions from $[0,1]$
to $\mathbb{R}$. 
\end{theorem}

\begin{proof}
By letting $t:=s/T$, we get
\begin{equation}\label{FT}
F(\theta,T)=\mathbb{E}\left[e^{-\theta T\int_{0}^{1}e^{\sigma W_{tT}+\zeta t-\frac{1}{2}\sigma^{2}tT}dt}\right]
=\mathbb{E}\left[e^{-\theta T\int_{0}^{1}e^{Z_{tT}+\zeta t}dt}\right],
\end{equation}
where
\begin{equation}\label{eqn:Z}
dZ_{t}=-\frac{1}{2}\sigma^{2}dt+\sigma dW_{t},\qquad X_{0}=0,
\end{equation}
which is equivalent to:
$dZ_{t/\sigma^{2}}=-\frac{1}{2}dt+dB_{t}$, with $X_{0}=0$,
where $B_{t}=\sigma W_{t/\sigma^{2}}$ is a standard Brownian motion by the Brownian scaling property.

Let $Y_{t}:=Z_{t/\sigma^{2}}$.
From the large deviations theory for small time diffusions (see e.g. \cite{Varadhan1967}),
$\mathbb{P}(Y_{\cdot(\sigma^2 T)}\in\cdot)$ satisfies
a sample path large deviation principle on $L_{\infty}[0,1]$, 
the space of functions from $[0,1]$
to $\mathbb{R}$ equipped with the supremum norm topology, 
with the speed $1/(\sigma^{2}T)$ and the good rate function (we refer to the definition of large deviation principle
and good rate function to \cite{Dembo} and general background of large deviations theory to \cite{Dembo,VaradhanII})
$I(g)=\frac{1}{2}\int_{0}^{1}\left(g'(t)\right)^{2}dt$,
with $g(0)=0$ and $g\in\mathcal{AC}[0,1]$, the space of absolutely continuous functions from $[0,1]$
to $\mathbb{R}$ and $I(g)=+\infty$ otherwise.

Since we have $\theta = \frac{2b^2}{\sigma^2 T}$, we obtain from \eqref{FT} that 
\begin{equation}
F(\theta,T)=\mathbb{E}\left[e^{-\frac{2b^2}{\sigma^{2}T}\int_{0}^{1}e^{X_{tT}+\zeta t}dt}\right]
=\mathbb{E}\left[e^{-\frac{2b^2}{\sigma^{2}T}\int_{0}^{1}\exp\{Y_{t(\sigma^{2}T)}+\zeta t\}dt}\right] \,.
\end{equation}
By using the fact that $e^{-\frac{2b^2}{\sigma^{2}T}\int_{0}^{1}\exp\{Y_{t(\sigma^{2}T)}+\zeta t\}dt}$ is uniformly
bounded between $0$ and $1$
and the map $g\mapsto\int_{0}^{1}e^{g(t)+\zeta t}dt$ is continuous from $L_{\infty}[0,1]$ to $\mathbb{R}$, 
we can apply Varadhan's lemma \cite{Dembo} to obtain:
\begin{equation}\label{limit:first}
\lim_{\sigma^2 T\rightarrow 0}(\sigma^{2}T)\log F(\theta,T)
=\sup_{g\in\mathcal{AC}[0,1]: g(0)=0}\left\{-2 b^2 \int_{0}^{1}e^{g(t)+\zeta t}dt
- \frac{1}{2}\int_{0}^{1}\left(g'(t)\right)^{2}dt\right\} \,,
\end{equation}
where the supremum is taken over all functions $g:[0,1]\to\mathbb{R}$ which are absolutely
continuous and satisfy $g(0)=0$. 
By defining $h(t) = g(t) + \zeta t$,  the extremal problem in \eqref{limit:first} 
reproduces (\ref{JB}) and hence completes the proof.
\end{proof}

The variational problem (\ref{JB}) giving the rate function $J_B(b,\zeta)$ is identical with the variational problem appearing in Theorem 3 in \cite{AsianDiscrete}, where it
gives the Lyapunov exponent associated with the moment generating function of the
discrete sum of a geometric Brownian motion.

This variational problem was solved completely in \cite{AsianDiscrete}, and the solution was reduced to the solution of a calculus problem in Proposition 4 of this paper. We quote the solution in the notations used here, 
adding the explicit condition on $\zeta,b$ distinguishing the two cases.
The result of Proposition 4 of \cite{AsianDiscrete} is mapped to our case by substituting
$a\mapsto 2b^2, b\mapsto 1, \rho \mapsto \zeta$.

\begin{proposition}
\label{prop:R}
The rate function of Theorem~\ref{Thm:1} is given explicitly by 
$J_B(b,\zeta) = - 2b^2 R(b,\zeta)$, which is defined as follows.

i) For $b \leq \frac{\zeta}{2+\zeta}$ we have
\begin{align*}
R(b,\zeta) &= 1 + \sinh^2(\delta/2) \left( 1 + \frac{\zeta(\zeta-4)}{\delta^2}\right)
- (2-\zeta)\frac{\sinh\delta}{\delta} \\
&\qquad\qquad\qquad + \frac{1}{b^2} \zeta \log\left( \cosh(\delta/2) + \frac{\zeta}{\delta} 
\sinh(\delta/2) \right) - \frac{\zeta^2}{2b^2}\,, \nonumber
\end{align*}
where 
$\delta\in [0,\zeta]$ is the solution of the equation
\begin{equation}\label{eq1}
\zeta^2 - \delta^2 = 4b^2 \Big( \cosh(\delta/2) + \frac{\zeta}{\delta} \sinh(\delta/2) \Big)^2\,.
\end{equation}

ii) For $b \geq \frac{\zeta}{2+\zeta}$ we have
\begin{equation*}
R(b,\zeta) = 1- \sin^2\xi \left( 1 + \frac{\zeta(4 - \zeta)}{4\xi^2} \right)
+ \frac{\zeta - 2}{2\xi} \sin(2\xi) + 
\frac{\zeta}{b^2} \log \left( \cos\xi + \zeta \frac{\sin \xi}{2\xi}
\right) - \frac{\zeta^2}{2b^2}\,, 
\end{equation*}
where $\xi$ is the unique solution $\xi \in (0,\frac{\pi}{2})$ of the equation
\begin{equation}\label{eq2}
2\xi^2 \left(4\xi^2 + \zeta^2\right) = 2b^2 \left(2\xi \cos\xi + \zeta \sin\xi\right)^2\,.
\end{equation}
\end{proposition}

We note that in Proposition~\ref{prop:R}, 
for $b = \frac{\zeta}{2+\zeta}$, the two cases i) and ii) give the common result
$R(b,\zeta) = -1 + \zeta - \frac12 (2+\zeta)^2 +
\frac{1}{\zeta}(2+\zeta)^2 \log(1+\zeta/2)$.

\subsection{Limiting case $a=0$} 
\label{sec:a0}
The solution simplifies greatly in the driftless GBM case $a=0$ which corresponds to $\zeta=0$. 
For $\zeta=0$ the rate function $J_B(b,0)$ is given by case (ii)
of Proposition~\ref{prop:R} (see also Corollary 5 in \cite{AsianDiscrete}):
\begin{equation}\label{2}
J_B(b,0) 
= 2 b^2 \left( \frac{\sin 2\lambda}{\lambda} - \cos^2 \lambda \right)\,,
\end{equation}
where $\lambda$ is the solution of the equation
\begin{equation}\label{lameq}
\frac{\lambda^2}{\cos^2\lambda} = b^2 \,.
\end{equation}

We give the asymptotics of $J_B(b,0)$ for small and large $b$, which is convenient for efficient numerical evaluation.

\begin{proposition}
The rate function $J_B(b,0)=2b^2 R(b,0)$ has the following asymptotic expansions:

i) small-$b$ asymptotics. As $b\to 0$ we have
\begin{equation}\label{Rexp}
R(b,0) = 1 - \frac13 b^2  +\frac{4}{15} b^4 - \frac{92}{315} b^6 + \frac{1072}{2835} b^8
+ O\left(b^{10}\right).
\end{equation}

ii) large-$b$ asymptotics. As $b\to \infty$ we have
\begin{equation}\label{Rexplargeb}
R(b,0) = \frac{2}{b} - \frac{\pi^2}{16b^2} - \frac{\pi^2}{8b^3} + O\left(b^{-4}\right).
\end{equation}
\end{proposition}

\begin{proof}
i) As $b\to 0$ the solution of (\ref{lameq}) approaches $\lambda\to 0$, and is expanded iteratively as
$\lambda^{(j)} = b\cos\left(\lambda^{(j-1)}\right)$
starting with $\lambda^{(0)}=0$. Substituting the result into (\ref{2}) reproduces (\ref{Rexp}).

ii) As $b \to \infty$ the solution of the equation (\ref{lameq}) approaches 
$\lambda \to \frac{\pi}{2}$ from below. Denote $\lambda = \frac{\pi}{2}-\varepsilon$,
and invert (\ref{lameq}) written as
$\frac{\sin \varepsilon}{\pi/2-\varepsilon} = 1/b $.
This gives  an expansion in $1/b$ for $\varepsilon$:
\begin{equation}\label{bexp}
\varepsilon = \frac{\pi}{2b} - \frac{\pi}{2b^2} + \frac{\pi(24+\pi^2)}{48b^3} + O\left(b^{-4}\right)\,.
\end{equation}
Substituting into (\ref{2}) and expanding in $1/b$ reproduces (\ref{Rexplargeb}).
\end{proof}

The function $R(b,0)$ appears also in the short maturity expansion of the at-the-money (ATM)
implied volatility in the $\beta=1$ SABR model in the combined small vol-of-vol 
and large volatility limit, see Proposition 23 in \cite{SABRdiscrete}. 
An examination of the singularities of the function $R(b,0)$ in the $b$ complex plane shows that the series expansion (\ref{Rexp}) has a finite convergence radius. By Proposition 2 in \cite{SABRconv}, the series for $R(b,0)$ converges for
\begin{equation}\label{Rconv}
|b| < R_b = \frac{y_0}{\cosh y_0} = 0.662743\,,
\end{equation}
where $y_0=1.19968$ is the positive solution of the equation $y \tanh y = 1$.

\begin{figure}
    \centering
   \includegraphics[width=3in]{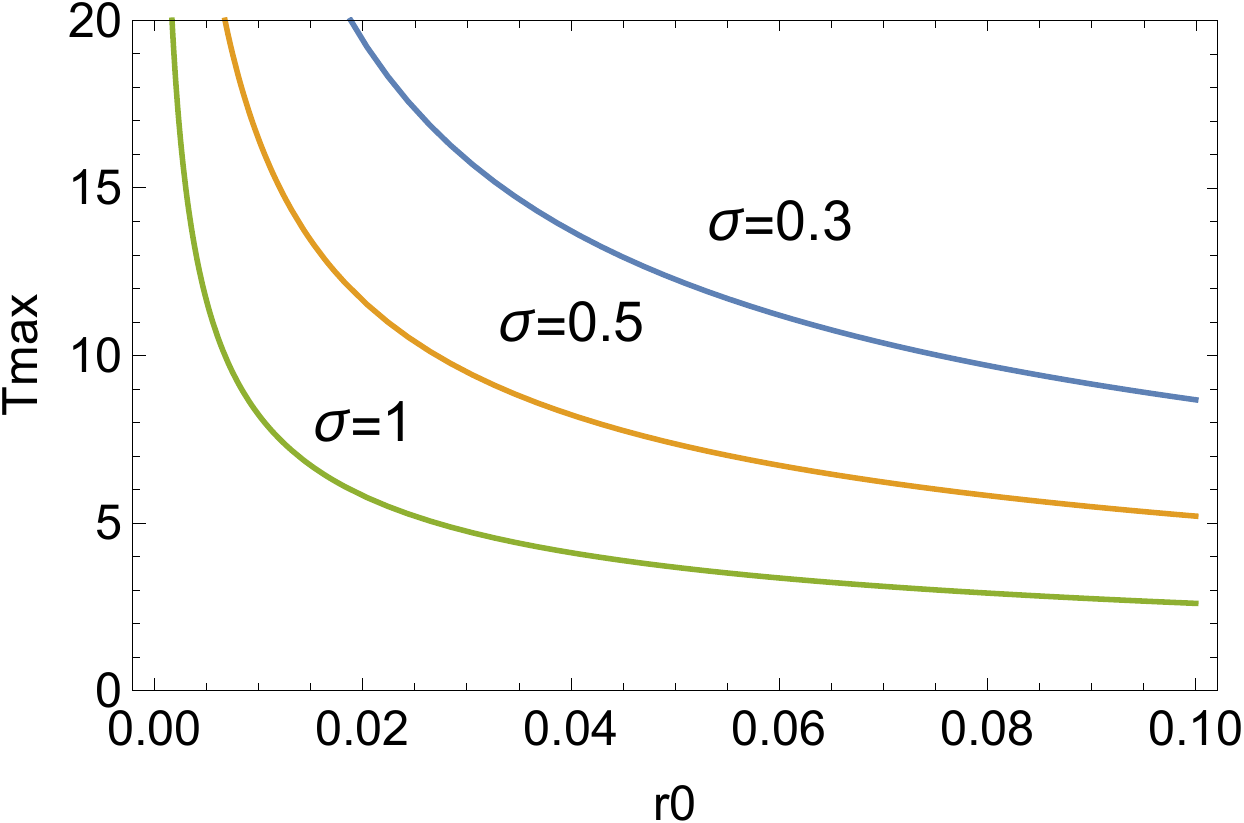}
    \caption{ Maximum maturity $T_{\rm max}(r_0,\sigma)$ vs $r_0$ for several values of $\sigma$, for which the series expansion (\ref{Rexp}) of $R(b,0)$ converges. }
\label{Fig:Tmax}
 \end{figure}
 
In the context of bond pricing in the Dothan model, the condition (\ref{Rconv}) can 
be put into a more explicit form as $\sigma^2 r_0 T^2 < 2R_b^2 = 0.582$. 
We show in Figure~\ref{Fig:Tmax} the range of the model parameters where this
condition is satisfied. The curves in this figure 
show the maximum maturity $T_{\rm max}(r_0,\sigma)$ vs $r_0$ for several
values of $\sigma = \{0.3, 0.5, 1.0\}$. The maximum maturity for which convergence holds decreases with $\sigma$ and $r_0$. 

Outside of the circle of convergence the series expansion (\ref{Rexp}) cannot be used,
and the exact result in (\ref{2}) must be used. However we emphasize that the asymptotic result $R(b,\zeta)$ exists and is well behaved for all real values of $b$, not only within the region of convergence.

\section{Large deviations for the time-integral of the gBM}\label{sec:LDP}

The asymptotic result of the previous section is related to another application in mathematical finance: the large deviations property of the time average of the geometric Brownian motion. Denoting $X_T := \int_0^T e^{\sigma W_t + (a-\frac12 \sigma^2) t} dt$, 
it was shown in \cite{SIFIN}
that $\mathbb{P}( X_{T}/T \in \cdot )$ satisfies a large deviation principle on 
$\mathbb{R}$ with speed $1/T$ and the rate function $\frac{1}{\sigma^2} J_{\rm{BS}}(\cdot)$, 
where $J_{\rm{BS}}(x)$ is given in explicit form in Proposition 12 in \cite{SIFIN}.
As is shown in \cite{SIFIN}, this result implies a short maturity asymptotics for out-of-the-money (OTM) Asian options in the Black-Scholes model where $a:=r-q$ with $r$ denoting the interest rate and $q$ the dividend yield, 
and neither $r$ or $q$ appears in the asymptotic result,
which is a general feature of the short-maturity limit $T\to 0$ at fixed $r$ and $q$. 

We study here the large deviations for $X_T$ in a different limit: 
$\sigma^{2}T\to 0$ at fixed $aT = \zeta$ with $a:=r-q$. 
The advantage of this limit is that it includes interest rates effects at leading order 
in the short maturity expansion. The asymptotic regime
is valid when $\sigma^{2}T$ is small, which is typically the case in practice. 

\begin{theorem}\label{Thm:2}
$\mathbb{P}( X_{T}/T \in \cdot )$ satisfies a large deviation
principle as $\sigma^2 T\to 0$ at fixed $\zeta := a T$
with speed $1/(\sigma^2 T)$ and rate function
\begin{equation}
I_{\rm BS}(x) = \inf_{h\in\mathcal{AC}[0,1]: h(0)=0, \int_0^1 e^{h(y)} dy = x}
\frac12 \int_0^1 \left( h'(y) - \zeta \right)^2 dy \,,
\end{equation}
where $\mathcal{AC}[0,1]$ denotes the space of absolutely continuous functions from $[0,1]$
to $\mathbb{R}$. 
\end{theorem}

\begin{proof}
Following a similar argument as in the proof of Theorem~\ref{Thm:1}, we have
\begin{equation}
\frac{X_{T}}{T}=\frac{1}{T}\int_0^T e^{\sigma W_t + (a-\frac12 \sigma^2) t} dt
=\int_{0}^{1}e^{Z_{tT}+\zeta t}dt=\int_{0}^{1}e^{Y_{t\sigma^{2}T}+\zeta t}dt,
\end{equation}
in distribution where $Z_{t}$ is defined in \eqref{eqn:Z} and $Y_{t}:=Z_{t/\sigma^{2}}$ satisfies
a sample path large deviation principle on $L_{\infty}[0,1]$, 
the space of functions from $[0,1]$
to $\mathbb{R}$ equipped with the supremum norm topology, 
with the speed $1/(\sigma^{2}T)$ and the good rate function 
$I(g)=\frac{1}{2}\int_{0}^{1}\left(g'(t)\right)^{2}dt$
with $g(0)=0$ and $g\in\mathcal{AC}[0,1]$, the space of absolutely continuous functions from $[0,1]$
to $\mathbb{R}$ and $I(g)=+\infty$ otherwise.
Since the map $g\mapsto\int_{0}^{1}e^{g(t)+\zeta t}ds$ is continuous from $L_{\infty}[0,1]$ to $\mathbb{R}$, 
we can apply contraction principle from large deviations theory \cite{Dembo} to 
conclude that $\mathbb{P}(X_{T}/T\in\cdot)$ satisfies a large deviation principle
with speed $1/(\sigma^2 T)$ and rate function
\begin{equation}
I_{\rm BS}(x) = \inf_{g\in\mathcal{AC}[0,1]: g(0)=0, \int_0^1 e^{g(t)+\zeta t} dt = x}
\frac12 \int_0^1 \left( g'(t)\right)^2 dt \,.
\end{equation}
By introducing $h(t) := g(t) + \zeta t$, we complete the proof.
\end{proof}

Using Theorem~\ref{Thm:2}, one can obtain the out-of-the-money (OTM)
asymptotics for Asian call and put options.
Denote the Asian call and put option prices as: 
$C(T):=e^{-rT}\mathbb{E}[(A_{T}-K)^{+}]$
and $P(T):=e^{-rT}\mathbb{E}[(K-A_{T})^{+}]$,
where $A_{T}:=\frac{1}{T}\int_{0}^{T}S_{t}dt$
with $S_{t}=S_{0}e^{(r-q)t+\sigma W_{t}-\frac{1}{2}\sigma^{2}t}$.
We have the following result which is the analog of Theorem 2 in \cite{SIFIN},
improved by keeping terms of $O((rT)^n)$ to all orders.

\begin{corollary}\label{corr:3.2}
(i) When $K>S_{0}$, $\lim_{\sigma^{2}T\rightarrow 0}(\sigma^{2}T)\log C(T)=-I_{\rm BS}(K/S_{0})$.

(i) When $K<S_{0}$, $\lim_{\sigma^{2}T\rightarrow 0}(\sigma^{2}T)\log P(T)=-I_{\rm BS}(K/S_{0})$.
\end{corollary}

\begin{proof}
The corollary follows from Theorem~\ref{Thm:2} by using a similar argument
as in the proof of Theorem~2 in \cite{SIFIN}. 
\end{proof}

The variational problem in Theorem~\ref{Thm:2} is identical to the variational problem appearing in Proposition~6 in \cite{AsianDiscrete}. This variational problem was solved in closed form and the solution is given in Proposition 9 in \cite{AsianDiscrete}. This solution is mapped to our case by the substitutions $2\beta\mapsto 1, \rho\mapsto \zeta, S_0\mapsto 1$.

We get the following result for $I_{\rm{BS}}(x)$, which is the continuous-time counterpart of the discrete-time result of \cite{AsianDiscrete}.

\begin{proposition}

i) For $x \geq 1 + \frac12 \zeta$ we have
\begin{equation*}
I_{\rm{BS}}(x) =\frac12\left(\delta^2 - \zeta^2\right) \left( 1 - \frac{2\tanh(\delta/2)}{\delta + \zeta \tanh(\delta/2)} \right) 
- 2\zeta \log \left(
\cosh(\delta/2) + \zeta \frac{\sinh(\delta/2)}{\delta} \right) + \zeta^2\,, 
\end{equation*}
where $\delta$ is the solution of the equation
$\frac{\sinh\delta}{\delta} + 2 \zeta \frac{\sinh^2(\delta/2)}{\delta^2} = x$.

ii) For $0 < x \leq 1 +\frac12 \zeta$ we have
\begin{equation*}
I_{\rm{BS}}(x)= 2\left(\xi^2 + \frac14 \zeta^2\right) \left( \frac{\tan\xi}{\xi + \frac12 \zeta \tan \xi}
- 1 \right) - 2\zeta \log\left( \cos\xi + \frac12 \zeta \frac{\sin\xi}{\xi} \right) + \zeta^2\,,
\end{equation*}
where $\xi \in (0,\frac{\pi}{2}) $ is the solution of the equation
$\frac{\sin(2\xi)}{2\xi} \Big( 1 + \frac12 \zeta \frac{\tan\xi}{\xi} \Big) = x$.
\end{proposition}

The properties of $I_{\rm{BS}}(x)$ are studied in detail in Section 4.1 of \cite{AsianDiscrete}.
We mention here only two properties.
i) The rate function $I_{\rm{BS}}(x)$ vanishes at $x=\frac{1}{\zeta}(e^\zeta-1)$. 
ii) For $\zeta = 0$ we have $I_{\rm{BS}}(x) = J_{\rm{BS}}(x)$ where $J_{\rm{BS}}(x)$ is the rate function for OTM Asian options studied in \cite{SIFIN}.

Corollary~\ref{corr:3.2} can be used to obtain an approximation for Asian option prices similar to the approach followed in Section 4.2 of \cite{SIFIN}. Under this approximation, Asian options can be priced as European options with the same
strike and maturity and an equivalent log-normal volatility $\Sigma_{\rm{LN}}(K,S_0)$ 
given by 
\begin{equation}
\Sigma_{\rm{LN}}^2(K,T) = \frac{\log^2 (K/A_{\rm{fwd}})}{2I_{\rm{BS}}(K/S_0)} \,.
\end{equation}
This reduces to the result of Proposition 18 in \cite{SIFIN} in the $\zeta=0$ limit, and improves it by taking into account interest rates effects. 
Numerical tests of this
improved approximation performed in Sec.~4.2 of \cite{AsianDiscrete} demonstrate good agreement with precise benchmarks, which is better than that given by the
asymptotic result in \cite{SIFIN} which neglects interest rates effects. The asymptotic result gives an alternative to other proposed pricing methods for Asian options, such as the spectral approach \cite{Linetsky}, the Laplace transform method \cite{Cai2012,Cai2015}, the small-time expansion method \cite{Cai2014}.

\section{Numerical tests for bond pricing in the Dothan model}\label{sec:numerical}

The asymptotic result of Proposition~\ref{prop:R} can be used to obtain an approximation for the bond prices in the Dothan model
\begin{equation}\label{Basympt}
B_{\rm asympt}(T) = e^{-r_0 T R(b,\zeta)},
\end{equation}
where $R(b,\zeta)$ is given by Proposition~\ref{prop:R}. In the limiting case $\zeta=0$ this simplifies further as shown in \eqref{2}.
In this section we present tests of this approximation under several scenarios.

\textbf{Scenario 1.} We start by considering scenarios with $a=0$. 
An exact solution for $B(T):= P_{0,T}$ was given in \cite{Dothan} and is represented as a double integral. For $a=0$ this 
reduces to a single integral
\begin{equation}\label{exact}
B(T) = \sqrt{y} \int_0^\infty  \sin(2 \sqrt{y} \sinh z )
\left[ e^{-z} \mbox{Erfc} \left(\frac{s - 2 z}{2\sqrt{s}}\right) -
e^z  \mbox{Erfc} \left(\frac{s + 2 z}{2\sqrt{s}}\right) \right]dz +
2 \sqrt{y} K_1(2\sqrt{y}) \,, 
\end{equation}
where $K_{1}$ is the modified Bessel function of order 1 with
$y := \frac{2r_0}{\sigma^2}$ and
$s = \frac{\sigma^2}{2}(T-t)$.

The direct numerical evaluation of the integral in
(\ref{exact}) becomes unstable for large $y$ due to fact that the integrand is rapidly oscillating. We found it convenient to add and
subtract in the square brackets the term $2e^{-z}$. The term proportional
to $+2e^{-z}$ can be evaluated in closed form using the relation
$\int_0^\infty dz e^{-z} \sin (a \sinh z) = \frac{1}{a} - K_1(a)$,
which gives
\begin{align}\label{exact1}
B(T) = \sqrt{y} \int_0^\infty  \sin(2 \sqrt{y} \sinh z ) 
\cdot
\left[ e^{-z} \mbox{Erfc} \left(\frac{s - 2 z}{2\sqrt{s}}\right) -
e^z  \mbox{Erfc} \left(\frac{s + 2 z}{2\sqrt{s}}\right)  - 2 e^{-z} \right]dz  +1. 
\end{align}
The integrand in this expression is still oscillatory, but its amplitude
falls off much faster, and the calculation of the integral is more stable. 
We used (\ref{exact1}) for the numerical evaluation of $B(T)$.

\begin{figure}
\centering
\includegraphics[width=1.5in]{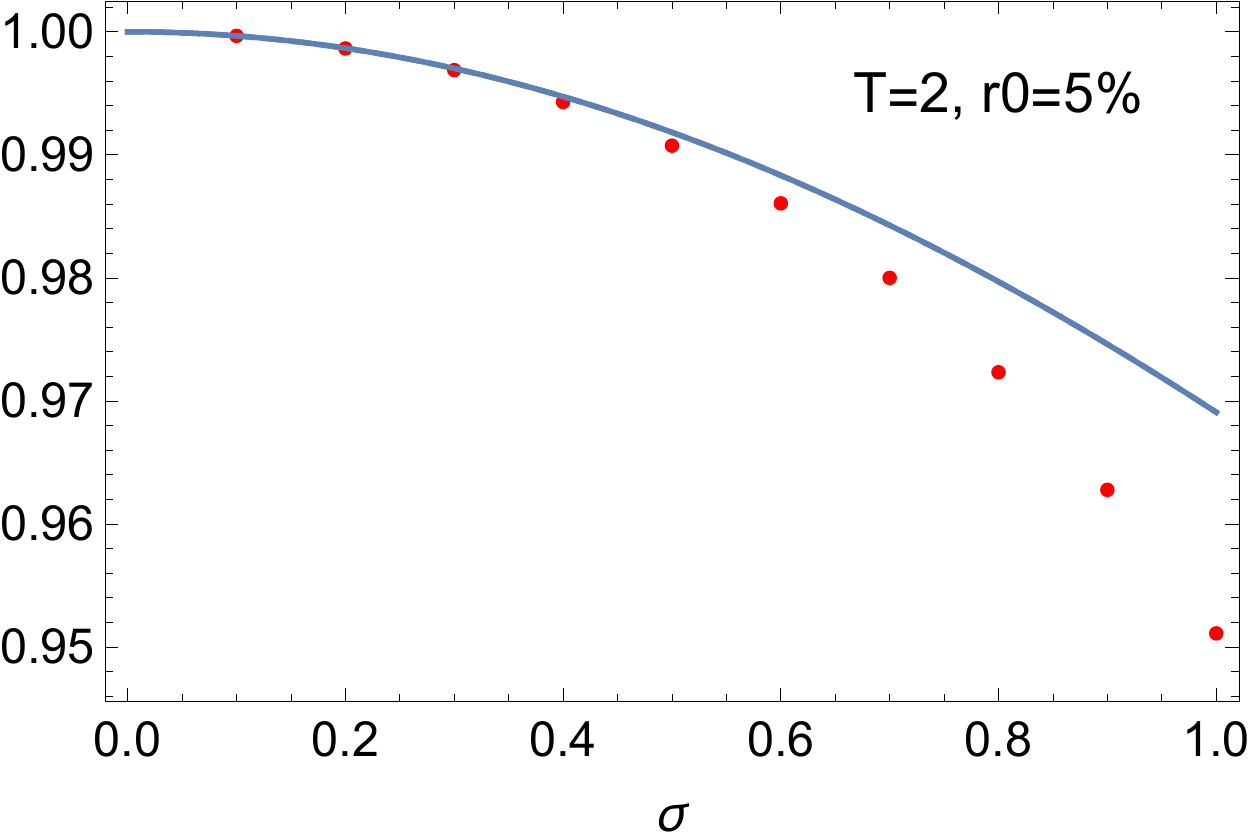}
\includegraphics[width=1.5in]{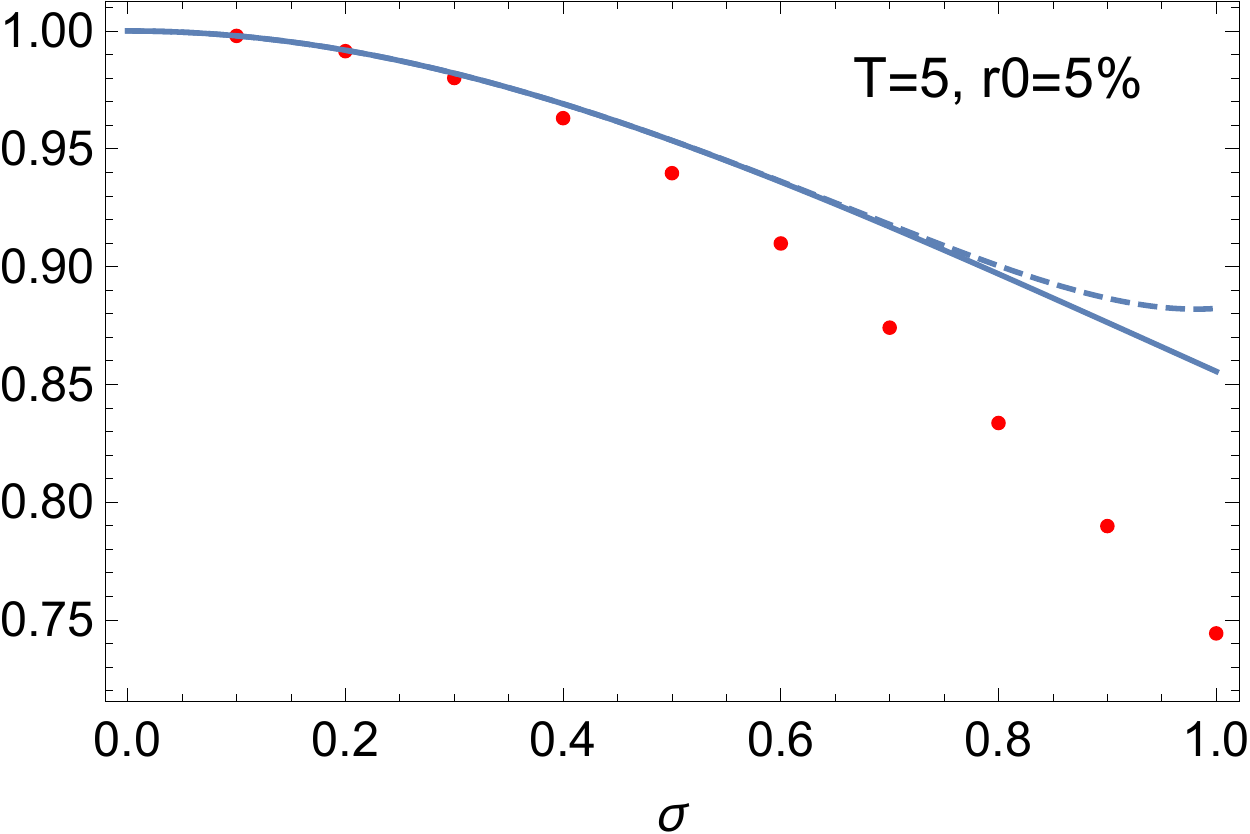}
\includegraphics[width=1.5in]{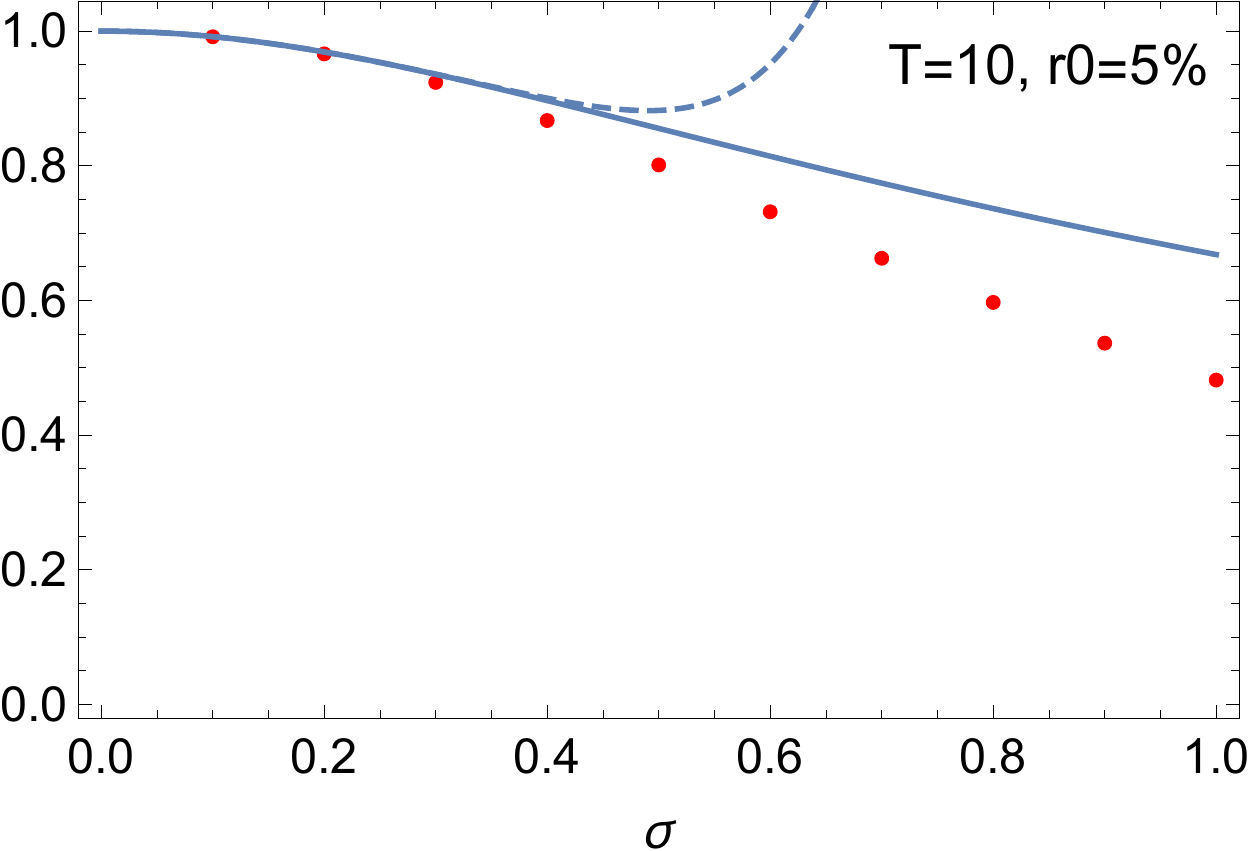}
\includegraphics[width=1.5in]{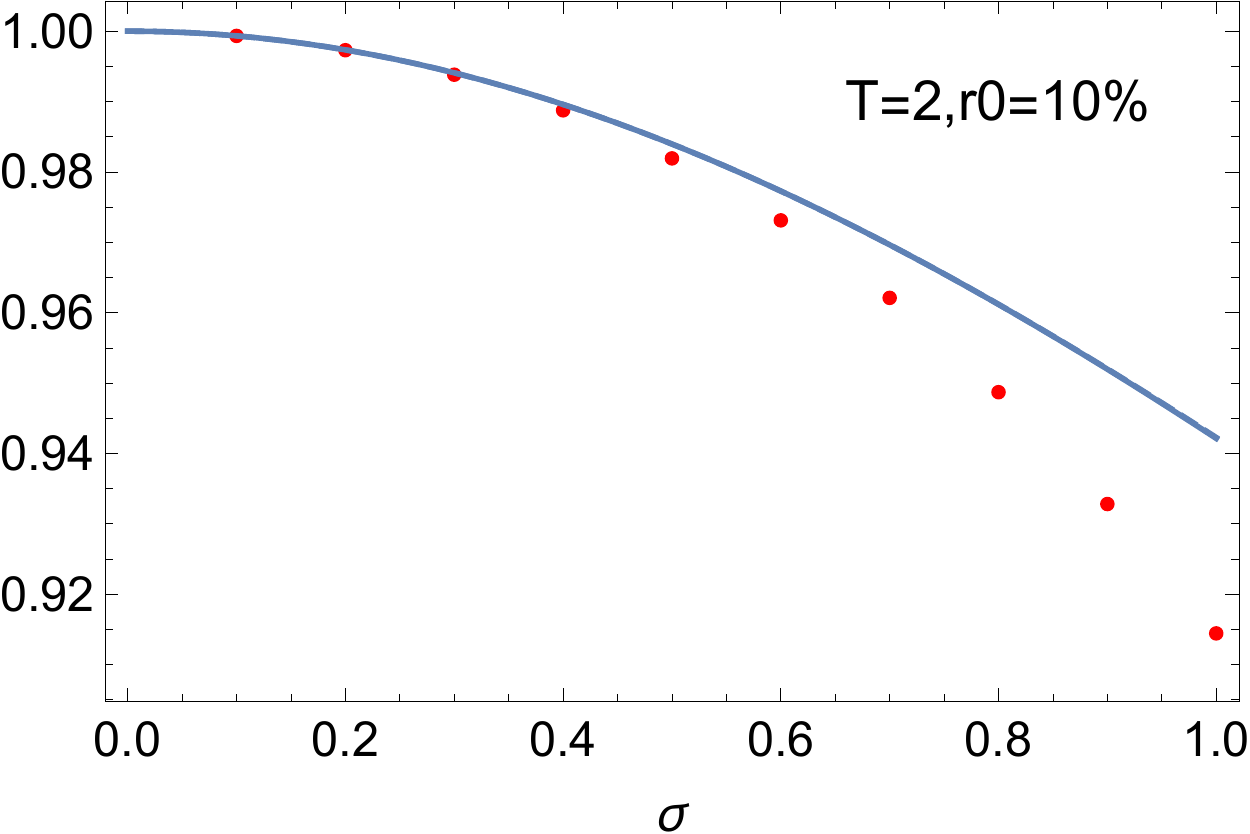}
\includegraphics[width=1.5in]{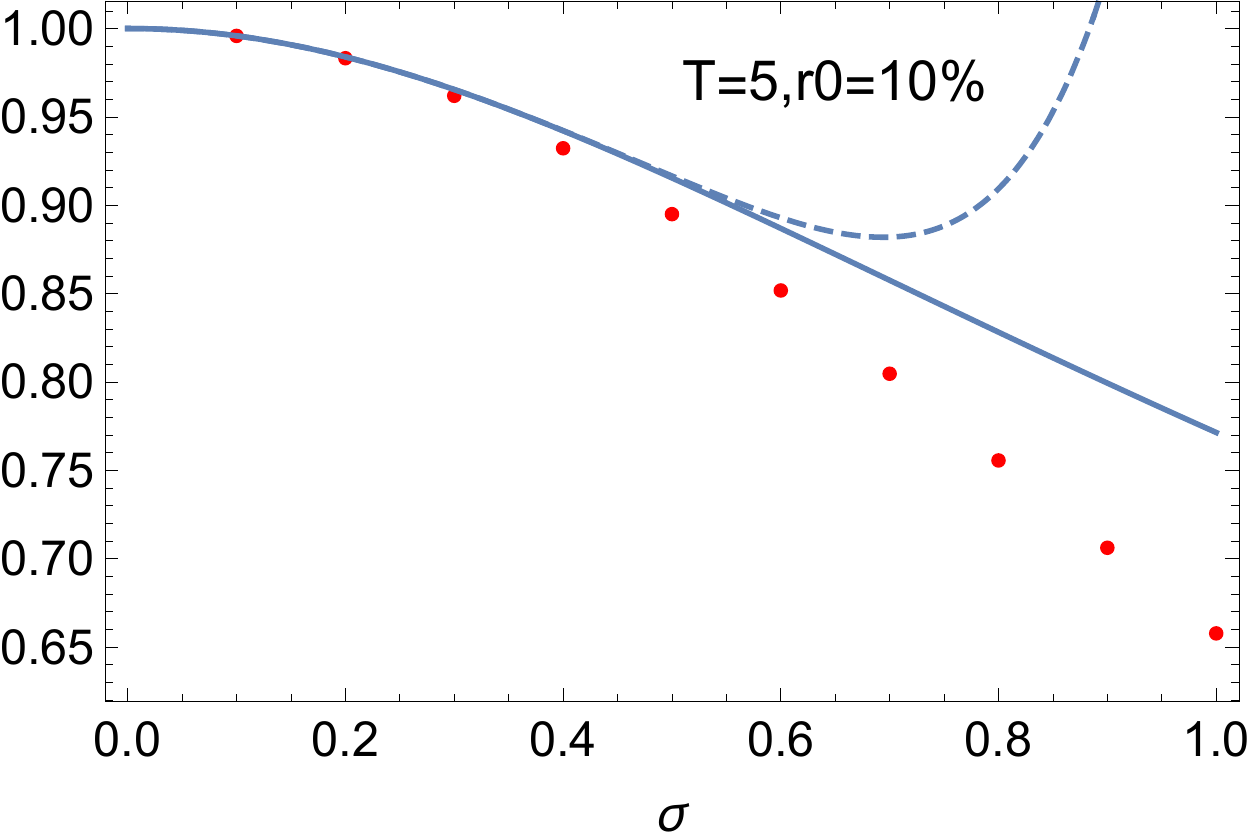}
\includegraphics[width=1.5in]{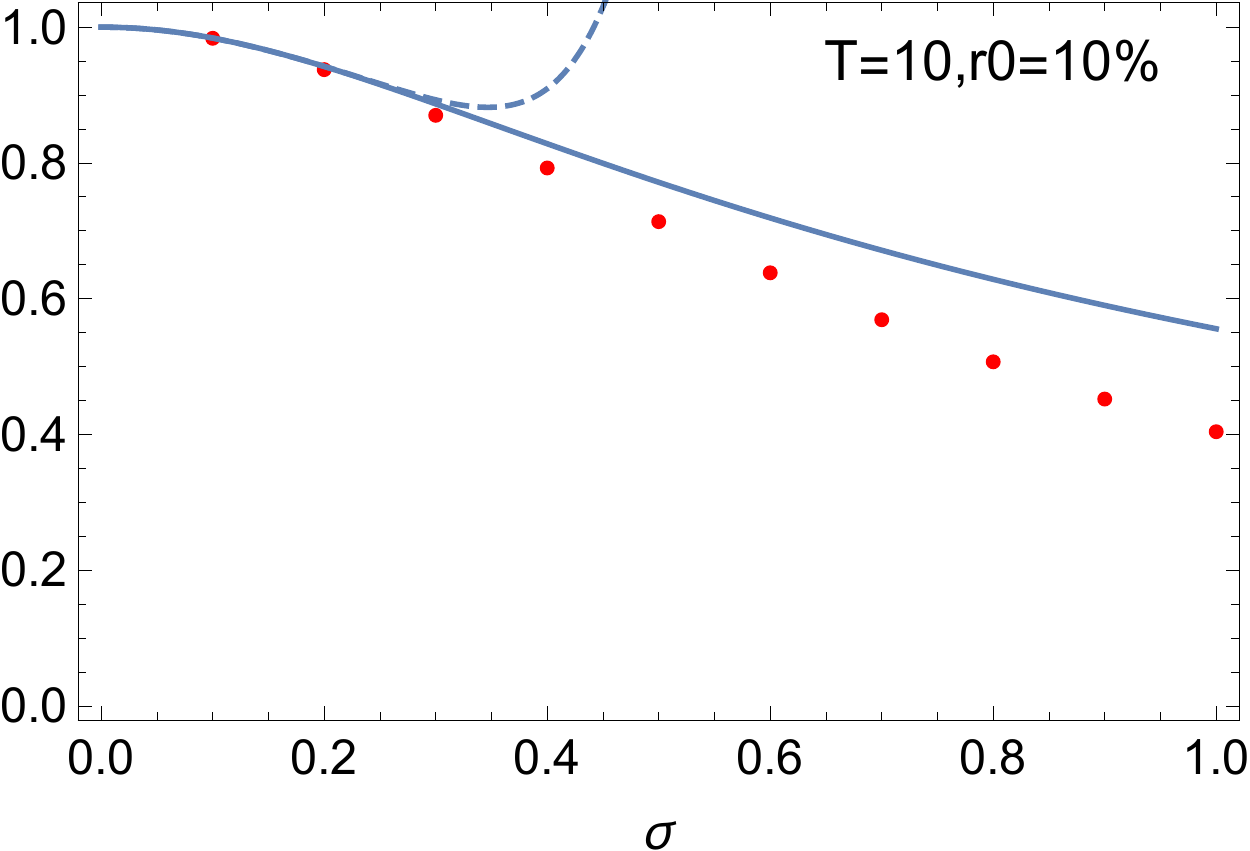}
\caption{Tests with $a=0$. Plots of $-\frac{1}{r_0T}\log B(T)$ vs $\sigma$ at fixed $T,r_0$. Solid curves: asymptotic result, dots: exact numerical evaluation of $B(T)$, dashed curve: series expansion (\ref{Rexp}) of the asymptotic result, keeping terms up to $O(b^8)$.}
\label{Fig:scenario1}
 \end{figure}

Figure~\ref{Fig:scenario1} shows the results for
$-\frac{1}{r_0 T} \log B(T)$ obtained by evaluation from (\ref{exact1}) (red dots),
comparing them with the asymptotic result of Proposition~\ref{prop:R} (solid blue curves).
The dashed blue curves show the
series expansion of the asymptotic result, Eq.~(\ref{Rexp}), keeping the first 8 terms.

The upper three plots in Figure~\ref{Fig:scenario1} correspond to a moderate interest rates regime $r_0=5\%$ and the lower plots to a high interest rates regime $r_0=10\%$. Selected numerical evaluations with $r_0=10\%$ are shown in 
Table~\ref{Table:1}.

From an examination of these tests we make a few observations:

i) The asymptotic result is most precise at small volatilities $\sigma$ and small maturities $T$. As either of these parameters increases, the agreement of the asymptotic result with the exact values worsens. Still, the asymptotic result gives a reasonably good approximation, better than 1\%, for all volatilities less than 20\%,
at all maturities less than 10Y, which corresponds to many cases of practical 
interest.

ii) As the interest rate $r_0$ increases, the agreement of the asymptotic result 
with the exact result improves, as expected from the scaling of the model parameters 
assumed in the asymptotic limit considered here. 

iii) The series expansion (\ref{Rexp}) truncated to a finite order explodes at a certain threshold to unphysical values ($R(b,0)$ becomes larger than 1, or smaller than 0). 
The explosion is to $+(-)\infty$ when the series is truncated to even/odd order,
corresponding to the sign of the last term included. 
However, the exact asymptotic result (\ref{2}) does not show any explosion.
As explained above, the failure of the series expansion is due to its finite convergence
radius, and not to a failure of the asymptotic expansion itself, which remains well behaved over the entire range of parameter values. 

Overall, the error of the asymptotic expansion is below 3.5\%
for volatilities below $\sigma=0.4$ and maturities up to 10 years, which covers 
a wide range of the parameters relevant for practical applications.

\begin{table}[t!]
\caption{\label{Table:1} 
Numerical evaluation of $R(T) := -\frac{1}{T} \log B(T)$ vs $\sigma$ at 
fixed $T,r_0=0.1$ and $a=0$.  
The fourth column shows the exact result obtained
by numerical evaluation using (\ref{exact1}), and the last column shows the
asymptotic result from Proposition~\ref{prop:R}.}
\begin{center}
\begin{tabular}{|cc|c|cc|}
\hline
$T$ & $\sigma$ & $B(T)$ & $R(T)$ &  $R_{\rm asympt}(T)$ \\
\hline
\hline
1   & 0.1  & 0.904853 & 9.998\% & 9.998\% \\
1   & 0.2  & 0.904898 & 9.993\% & 9.993\% \\
1   & 0.3  & 0.904976 & 9.985\% & 9.985\%\\
1   & 0.4  & 0.905087 & 9.972\% & 9.973\%\\
1   & 0.5  & 0.905235 & 9.956\% & 9.959\%\\
\hline
5   & 0.1  & 0.607799  & 9.958\% & 9.959\%  \\
5   & 0.2  & 0.611650  & 9.832\% & 9.840\%  \\
5   & 0.3  & 0.618183  & 9.619\% & 9.655\% \\
5   & 0.4  & 0.627431  & 9.322\% & 9.421\% \\
5   & 0.5  & 0.639230  & 8.950\% & 9.155\% \\
\hline
10  & 0.1  & 0.373968  & 9.836\% & 9.839\% \\
10  & 0.2  & 0.391646  & 9.374\% & 9.421\% \\
10  & 0.3  & 0.418920  & 8.701\% & 8.869\% \\
10  & 0.4  & 0.452708  & 7.925\% & 8.282\% \\
10  & 0.5  & 0.489961  & 7.134\% & 7.714\% \\
\hline
\end{tabular}
\end{center}
\end{table}

\textbf{Scenario 2.} 
We present also numerical tests with $a\neq 0$ using the benchmark scenarios of \cite{PU,Kim}, which considered the pricing of a zero
coupon bond in the Dothan model with $(a-\frac12\sigma^2) = 0.045$ which 
corresponds to $a=0.09$ in our notations. 
The initial interest rate is $r_0=0.06$ and volatility $\sigma =0.3$. 

In these papers the zero coupon bond prices have been evaluated in two ways:
i) by Monte Carlo evaluation of the expectation (\ref{ZCBdef}) with importance 
sampling and control variate \cite{PU} and optimal change of drift \cite{Kim}, 
and ii) by evaluation of an alternative probabilistic representation of this 
quantity as an expectation of a function of a generalized hyperbolic secant
random variable \cite{PU}. The results of the two approaches are shown in 
Figures 3.2 and 3.4 of \cite{PU}, respectively.

\begin{table}[t!]
\caption{\label{Table:3} 
Zero coupon bond prices in the Dothan model under the scenario of \cite{PU}, comparing the asymptotic result $B_{\rm asympt}(T)$ with the exact evaluation of $B(T)$ from \cite{PU}
(last column).}
\begin{center}
\begin{tabular}{|c|cccc|}
\hline
$T$ & $\xi$ & $-\frac{1}{T} \log B_{\rm asympt}(T)$ & $B_{\rm asympt}(T)$ & $B(T)$ \\
\hline
\hline
1   &  0.030345  & 0.06272 & 0.939 & 0.939 \\
2   &  0.068373  & 0.06547 & 0.877 & 0.877 \\
3   &  0.112756  & 0.06821 & 0.814  & 0.815 \\
4   &  0.162295  & 0.07091 & 0.753  & 0.753 \\
5   &  0.215833  & 0.07354 & 0.692  & 0.693 \\
10   &  0.507276  & 0.08454 & 0.429 & 0.438 \\
15   &  0.777869  & 0.09113 & 0.255  & 0.275 \\
20   &  1.001668  & 0.09411 & 0.152  & 0.179 \\
\hline
\end{tabular}
\end{center}
\end{table}

For our tests we use the asymptotic result of Proposition~\ref{prop:R} to compute
$R(b)$, which is used to obtain $B_{\rm asympt}(T)$ using (\ref{Basympt}). For this test we use the scenario of \cite{PU}  ($r_0=0.06, 
\sigma=0.3, a=0.09$).
The asymptotic results are shown in columns 3 and 4 of Table~\ref{Table:3} for several values of the maturity up to $T=20$ years. The second column shows $\xi$,
the solution of the equation (\ref{eq2}) determining the rate function in Proposition~\ref{prop:R}. The last column shows the exact numerical evaluation of $B(T)$ obtained using the methods of \cite{PU}.
The approximation error of the asymptotic result is about 2\% at $T=10$ and increases to 15\% at $T=20$.

\section{Discussion and comparison with the literature}

In this note we derived an asymptotic result for the Laplace transform
of the integral of the geometric Brownian motion $F(\theta,T)$ in a new limit 
$\sigma^2 T\to 0$ at fixed $\sigma^2 \theta T^2$.
This result was applied to the pricing of zero coupon bond prices in the
Dothan model. For this case the limit includes the case of small volatility $\sigma$ at fixed
maturity $T$ and large interest rate $r_0$.
The asymptotic result can be used to obtain an efficient numerical
evaluation of the bond prices. We demonstrate good
agreement in these regimes with a numerical evaluation  
using an integral representation proposed by \cite{Dothan}. 
The method proposed
here requires only the solution of a non-linear equation and 
can be used in practical applications for a fast and precise evaluation.

Several authors presented asymptotic expansions of bond prices in models with 
log-normal rates, including the Dothan model. We discuss briefly  
the relation to our work.

Ref.~\cite{Stehlikova} derived a Taylor expansion of the log of the bond prices in maturity $T$, see also \cite{BSS2017} for a survey. The approach was applied to several one-factor short rate models with constant coefficients, including the Dothan model. Expressed in our notation, the first few terms in this expansion read (for simplicity we assume $a=0$)
\begin{eqnarray}
-\frac{1}{r_0 T} \log B(T) = 1 - \frac{1}{3!} \sigma^2 r_0 T^2 - 
\frac{1}{4!} \sigma^4 r_0 T^3 
- \left(  \frac{1}{5!} \sigma^6 r_0 -
\frac{1}{15} \sigma^4 r_0^2 \right) T^4 + O\left(T^5\right).\nonumber
\end{eqnarray}

Taking the limit $\sigma^2 T\to 0$ at fixed $b^2 = \frac12 \sigma^2 r_0 T^2$ this expansion becomes $1-\frac13 b^2 + \frac{4}{15} b^4 + \cdots$, which reproduces indeed the first few terms in the expansion (\ref{Rexp}). This shows that the 
asymptotic approximation (\ref{Basympt}) corresponds to summing a subset of the terms in the exact expansion, to all orders in $T$.
As shown in Sec.~\ref{sec:a0} this subset of the full series converges with a finite convergence radius. It would be interesting to study the convergence of the full series of $\log B(T)$ in powers of $T$. We mention \cite{Lewis1999} which studied the small-$T$ expansion of bond prices in several short rate models.

Several groups derived small volatility expansions for zero coupon bond prices in short rate models. 
\cite{Tourruco2007}
presented an expansion of zero coupon bonds in powers of volatility in a 
short rate model with short rate $r_t = r_0 (1 +\nu X_t)^{1/\nu}$ where $X_t$ is an Ornstein-Uhlenbeck (OU) process.  
This model recovers the Black-Karasinski model in the limit $\nu\to 0$, which in turn reduces to the Dothan model in the limit of positive mean-reversion. 
A similar expansion was proposed by \cite{AS2011}; their approach rescales simultaneously the mean-reversion level and the volatility of the OU process. 
\cite{SC2014} derived an expansion for the Arrow-Debreu functions $\psi(r_0,r_T,T) := \mathbb{E}[e^{-\int_0^T r_s ds } \delta(r_T-r_0)]$
using the so-called Exponent Expansion. This recovers the zero coupon bonds in the
Dothan model by an integration. 

Expansion-based approximations truncated to a finite order will typically fail as the maturity or volatility exceeds a certain value. This is similar to the failure of the expansion of the function $R(b,0)$ in powers of $b$ seen in the numerical tests in Figure~\ref{Fig:scenario1}. 
Our results suggest that these failures are not necessarily associated with the small-maturity or volatility limits considered, but are due to: i) truncation of a convergent series within its convergence radius, or ii) the finite convergence radius of the expansion. Using the full asymptotic result removes the singular behavior observed in the truncated series.

\subsection*{Acknowledgements} 
We are grateful to the Associate Editor and two anonymous referees for helpful comments and suggestions.
We thank Nicolas Privault and Wayne Uy for 
providing details of the numerical evaluations in their work \cite{PU}.
Lingjiong Zhu is partially supported by the grants NSF DMS-2053454, NSF DMS-2208303 and a Simons Foundation Collaboration Grant. 

\appendix

\section{Small interest rate ($r_{0}$) asymptotics}
\label{sec:small:r:0}

In this Appendix, we obtain an asymptotic expansion for the zero coupon bond
price $B(T):= P_{0,T}$ in the Dothan model for small interest rate $r_{0}$, where $P_{0,T}=\mathbb{E}\left[e^{-\int_{0}^{T}r_{s}ds}\right]$ with $r_{t}=r_{0}e^{\sigma W_{t}+(a-\frac{1}{2}\sigma^{2})t}$ follows
a geometric Brownian motion (gBM).
Ref.~\cite{Hansen2000} studied the numerical evaluation of
bond prices in the Dothan model by expansion of the expectation in powers of $r_0$. This expansion has the form $B(T) = \sum_{k=0} \frac{1}{k!}(-1)^m r_0^k m_k$, where $m_k$ is the $k-$th moment of the time integral of the gBM. The positive integer moments of the 
time integral of the gBM are known in closed form to all orders \cite{Yor,Dufresne2004}.

Numerical evaluation of the series expansion in powers of $r_0$ in \cite{Hansen2000}
shows good convergence for a few benchmark cases. However, the series expansion
$B(T) = \sum_{k=0} \frac{1}{k!} r_0^k m_k$ has zero convergence radius at $r_0=0$, and diverges for any $r_0>0$. The vanishing of the convergence radius follows by noting that $B(T)=\infty$ for any $r_0<0$. 
Thus the series expansion for $B(T)$ is strictly asymptotic.

The practical application of such series requires that they are truncated to an optimal truncation order $k_{\rm opt}$, see \cite{Boyd} for an overview. The optimal order can be determined by empirical numerical evaluation, as the order at which the next term in the series is minimal,
and in general will depend on $\sigma,T,a$.

We reformulate here the series expansion in $r_0$ as an exact limiting result in the
$r_0\to 0$ limit. We have the following proposition.

\begin{proposition}\label{prop:small:r:0}
We have the asymptotics:
\begin{equation}
\lim_{r_{0}\rightarrow 0}
\frac{1}{r_{0}^{2}}\left(B(T)-1+r_{0}\frac{e^{aT}-1}{a}\right)
=\frac{1}{a+\sigma^{2}}\left(\frac{e^{(2a+\sigma^{2})T}-1}{2a+\sigma^{2}}-\frac{e^{aT}-1}{a}\right).
\end{equation}
\end{proposition}

\begin{proof}
First, notice that for any $x\geq 0$, we have
\begin{equation}\label{two:bounds}
1-x+\frac{x^{2}}{2}-\frac{x^{3}}{6}\leq e^{-x}\leq 1-x+\frac{x^{2}}{2} \,.
\end{equation}

By Jensen's inequality, $\left(\frac{1}{T}\int_{0}^{T}e^{\sigma W_{s}+(a-\frac{1}{2}\sigma^{2})s}ds\right)^{3}\leq
\frac{1}{T}\int_{0}^{T}\left(e^{\sigma W_{s}+(a-\frac{1}{2}\sigma^{2})s}\right)^{3}ds$, which implies that
\begin{equation}\label{bound:3}
\mathbb{E}\left[\left(\int_{0}^{T}e^{\sigma W_{s}+(a-\frac{1}{2}\sigma^{2})s}ds\right)^{3}\right]
\leq
T^{2}\mathbb{E}\left[\int_{0}^{T}e^{3\sigma W_{s}+3(a-\frac{1}{2}\sigma^{2})s}ds\right]
=T^{2}\frac{e^{3(a+\frac{1}{2}\sigma^{2})T}-1}{3(a+\frac{1}{2}\sigma^{2})}.
\end{equation}
The first two moments are evaluated in closed form as $\mathbb{E}\left[\int_{0}^{T}e^{\sigma W_{s}+(a-\frac{1}{2}\sigma^{2})s}ds\right]=\frac{e^{aT}-1}{a}$
and\\ 
$\frac{1}{2}\mathbb{E}\left[\left(\int_{0}^{T}e^{\sigma W_{s}+(a-\frac{1}{2}\sigma^{2})s}ds\right)^{2}\right]=\frac{1}{a+\sigma^{2}}\left(\frac{e^{(2a+\sigma^{2})T}-1}{2a+\sigma^{2}}-\frac{e^{aT}-1}{a}\right)$.
The stated results follows by combining \eqref{bound:3}, \eqref{two:bounds} and the definition of $B(T)$.
\end{proof}

\begin{remark}
Proposition~\ref{prop:small:r:0} can be extended to 
arbitrarily high order in $r_0$ as:
$$\lim_{r_{0}\rightarrow 0}\frac{1}{r_{0}^{k}}\left(B(T)-\sum\nolimits_{j=0}^{k-1}\frac{(-1)^{j}}{j!}r_{0}^{j}m_{j}\right)=\frac{m_{k}}{k!}$$
for any arbitrary $k\in\mathbb{N}$, where $m_{k}$ is the $k$-th moment of $\int_{0}^{T}e^{\sigma W_{s}+(a-\frac{1}{2}\sigma^{2})s}ds$. 
\end{remark}

\section{Large maturity ($T$) asymptotics}\label{sec:large:T}

In this Appendix we discuss the asymptotics for the zero coupon bond
price $B(T)$ in the Dothan model for large maturity $T$. 
If $a<\frac12\sigma^2$, the bond price approaches a finite limit in the infinite maturity limit. The perpetual bond price is
\begin{equation}\label{largeT}
\lim_{T\to \infty} B(T) = B(\infty):= 
\frac{2}{\Gamma(1-\frac{2a}{\sigma^2})}
(2r_0/\sigma^2)^{\frac12 - \frac{a}{\sigma^2}}
K_{1-\frac{2a}{\sigma^2}}\left(2\sqrt{2r_0/\sigma^2} \right)\,,
\end{equation}
where $K_{\alpha}$ is the modified Bessel function of order $\alpha$.
This follows from the well-known result of \cite{Dufresne1990}: the time integral of the geometric Brownian motion converges in distribution as $T\to \infty$ 
to an inverse gamma distribution. The limiting result (\ref{largeT}) is the Laplace transform of the resulting inverse gamma distribution.

Let us consider the asymptotics of the exact perpetual bond price $B(\infty)$ 
as $\frac{2r_0}{\sigma^2}\to \infty$. 
Using the asymptotic expansion of the Bessel function of large argument
$K_\alpha(z) \sim \sqrt{\frac{z}{2\pi}} e^{-z}(1+O(z^{-2}))$ for any $\alpha>0$ gives 
\begin{equation}\label{Bperpetual}
B(\infty) = C e^{- 2 \sqrt{\frac{2r_0}{\sigma^2}}}\,.
\end{equation}

We show next that the asymptotic result of Proposition~1 in the main paper reproduces the exponential factor of the exact result (\ref{Bperpetual}), as expected, since the 
asymptotic limit $\sigma^2 T\to 0$ at fixed $b^2 = \frac12 \sigma^2 r_0 T^2$ corresponds to $\frac{2r_0}{\sigma^2} = \frac{2b^2}{(\sigma^2 T)^2} \to \infty$.
For simplicity we consider $a=0$. Using the large$-b$ expansion of the rate function 
$R(b,0) = \frac{2}{b} - \frac{\pi^2}{16b^2} - \frac{\pi^2}{8b^3} + O\left(b^{-4}\right)$ in Proposition~2 in the main paper, we have
\begin{eqnarray}
-\log B_{\rm asympt}(T) = r_0 T R(b,0) 
= - 2 \sqrt{\frac{2r_0}{\sigma^2}} + \frac{\pi^2}{2\sigma^2 T} + O\left(T^{-2}\right)\,,
\nonumber
\end{eqnarray}
which reproduces indeed the exponential factor in (\ref{Bperpetual}).

\bibliographystyle{alpha}
\bibliography{Dothan}

\end{document}